\documentclass[smallextended]{article}

\usepackage{amsmath}
\usepackage{amsfonts}
\usepackage{amssymb}

\usepackage{bbold}
\usepackage{comment}
\usepackage{subfigure}
\usepackage{tikz}
\usetikzlibrary{arrows,automata,matrix}
\usepackage{authblk}

\newcommand{\setofQS}{\varOmega}
\newcommand{\TOM}[1]{\mathcal{#1}}
\newcommand{\TOMel}[1]{\mathcal{#1}}
\newcommand{\R}{\mathbb{R}}
\newcommand{\HH}{\mathcal{H}}
\newcommand{\1}{\mathbb{1}}
\newcommand{\ket}[1]{\left|#1\right>}
\newcommand{\bra}[1]{\left<#1\right|}
\newcommand{\ketbra}[2]{\ket{#1}\!\!\bra{#2}}

\DeclareMathOperator{\tr}{tr}

\newtheorem{definition}{Definition}
\newtheorem{lemma}{Lemma}

\newtheorem{theorem}{Theorem}
\newtheorem{remark}{Remark}

\newtheorem{proof}{Proof}

\title{Quantum Hidden Markov Models based on Transition Operation Matrices}

\author[1]{Micha\l{} Cholewa} 
\author[1]{Piotr Gawron}
\author[1]{Przemys\l{}aw G\l{}omb}
\author[1,2]{Dariusz Kurzyk}

\affil[1]{Institute of Theoretical and Applied Informatics, Polish Academy of Sciences, Ba\l{}tycka 5, 44-100 Gliwice, Poland}
\affil[2]{Institute of Mathematics, Silesian University of Technology, Kaszubska 23, Gliwice 44-100, Poland}
\affil[]{\it\{mcholewa,gawron,przemg,dkurzyk\}@iitis.pl}

\date{10.02.2017}

\begin{document}
\maketitle

\begin{abstract}
In this work, we extend the idea of Quantum Markov chains [S. Gudder. Quantum
Markov chains. J. Math. Phys., 49(7), 2008] in order to propose Quantum Hidden
Markov Models (QHMMs). For that, we use the notions of Transition Operation
Matrices (TOM) and Vector States, which are an extension of classical stochastic
matrices and probability distributions. Our main result is the Mealy QHMM
formulation and proofs of algorithms needed for application of this model:
Forward for general case and Vitterbi for a restricted class of QHMMs. We show
the relations of the proposed model to other quantum HMM propositions and
present an example application.

\textbf{keywords}: Hidden Markov Models \and open quantum walks \and Transition Operation Matrices
\end{abstract}

\section{Introduction}
The most basic Markov model is a Markov chain, which can be defined as a
stochastic process with the Markov property. Formally, a Markov chain is a
collection of random variables $\{n_t, t\geq 0\}$ having the property that
$
P(\!n_{t+1}\!=\!S_{k_{t+1}}|n_1\!=\!S_{k_1},n_2\!=\!S_{k_2},...,n_t\!=\!S_{k_t}\!)\!=\!P(\!n_{t+1}\!=\!S_{k_{t+1}}|n_t\!=\!S_{k_t}\!),
$
where the values $\{S_1,...,S_T\}$ of $n_t$ are called states. They form the state space of the chain.
According to the Markov property, the current state of a chain is only dependent on
the previous state. Moreover, the state of a Markov chain is directly observed in
each step. Any Markov chain can be described by a directed graph called the state
diagram, where vertices are associated with states and each edge $(i,j)$
is labelled by the probability of going from $i$-th state to $j$-th state. The
information about Markov chain can be also represented by the initial state $\mathrm{\Pi}$
and the stochastic matrix called transition matrix $\mathrm{P}=[p_{ij}]$, such that
$p_{ij}=P(n_{t+1}=S_j|n_t=S_i)$. If we consider a Markov chain where
states are not observed directly, and these states generate symbols according to some
random variables, then we obtain a Hidden Markov Model (HMM). Hence, in the case of a 
Markov chain, the states correspond with observations, but for a HMM, the states
correspond with the random source of observations.

The classical Hidden Markov Model was introduced as a method of modelling signal
sources observed in noise. It is now extensively used, e.g. in speech and
gesture recognition or biological sequence analysis. Their popularity is a
result of their versatile structure, which is able to model wide variety of problems,
and effective algorithms that facilitate their application. The HMM is
related to three fundamental: given a~sequence of symbols of length 
$T$, $O=\left(o_1,o_2,\ldots,o_T\right)$, and a
HMM parametrized by $\lambda$,
\begin{enumerate}
\item Compute the $P(O|\lambda)$, probability that the sequence $O$ can be produced by a HMM $\lambda$.
\item Select the sequence of state indexes $N_T=\left(n_0, n_1, \dots, n_T\right)$ that
maximizes the probability $P(O|\lambda, N_T)$; in other words the most likely
state sequence in HMM $\lambda$ that produces $O$.
\item Adjust the model parameters $\lambda$ to maximize $P(O|\lambda)$.
\end{enumerate}
The above problems are solved, respectively, by the Forward, Vitterbi and
Baum-Welch algorithms. The effectiveness of those algorithms is based on
optimized procedure of computation, which uses a `trellis': a two dimensional
lattice structure of observations and states. This formulation is based on the
Markov property of model evolution and reduces the complexity from exponential
$\mathcal{O}(TN^T)$ to polynomial $\mathcal{O}(N^2T)$, where $T$ is the number
of observations and $N$ the number of model states \cite{rabiner1989tutorial}.

Depending on the formulation, there are two definitions of a Hidden Markov
Model: Mealy and Moore. In the former, the probability of next state being $n_{t+1}$
depends both on the current state $n_{t}$ and the generated output symbol $o_t$.
In the latter, the symbol generation is independent from state switch, i.e.
$P(n_{t + 1}=S_i|o_{t+1}=o, n_t=S_j) = P(n_{t + 1}=S_i|n_t=S_j)$. While the
expressive power of Moore and Mealy models is the same, i.e. a process can be
realized with Moore model if and only if it is realizable by Mealy model, the
minimal model order for the realization is lower in Mealy models
\cite{vanluyten2008equivalence}. In this work we focus only on Mealy models.

\subsection{Related work}\label{sec:related-work}
In this work we follow the scheme proposed by Gudder in \cite{gudder2008quantum}
and extend it in order to construct Quantum Hidden Markov Models (QHMMs). Gudder
introduced the notions of Transition Operation Matrices and Vector States, which
give an elegant extension of classical stochastic matrices and probability
distributions. These notions allow to define Markov processes that exhibit both
classical and quantum behaviour.

Below we review two areas of research most closely related to our work: open
quantum walks and Hidden quantum Markov models.

\paragraph{Open quantum walks}
In recent years a new sub-field of quantum walks has emerged. In series of
papers
\cite{petruccione2011open,attal2012open,sinayskiy2012efficiency,sinayskiy2012properties,attal2012opena,sinayskiy2013open}
Attal, Sabot, Sinayskiy, and Petruccione introduced the notion of Open Quantum
Walks. Theorems for limit distributions of open quantum random walks were
provided in \cite{konno2013limit}. In \cite{ampadu2014averaging} the average
position and the symmetry of distribution in the $SU(2)$ Open Quantum Walk is
studied.
The notion of open quantum walks is generalised to quantum operations of any
rank in \cite{pawela2015generalized} and analysed in \cite{sadowski2016central}.
In first of these two papers the notion of mean first passage time for a
generalised quantum walk is introduced and studied for class of walks on
Apollonian networks. In the second paper a central limit theorem for reducible
and irreducible open quantum walks is provided.
In a recent paper \cite{li2015quantum} authors introduce the notion of
hybrid quantum automaton -- an object similar to quantum hidden Markov model.
They use hybrid quantum automata and derived concepts in application to model
checking.

\paragraph{Quantum hidden Markov models}
Hidden quantum Markov models were introduced in \cite{monras2011hidden}. The
construction provided there by the authors is different from ours. In their work
the hidden quantum Markov model consists of a set of quantum operations
associated with emission symbols. The evolution of the system is governed by the
application of quantum operations on a quantum state. The sequence of emitted
symbols defines the sequence of quantum operations being applied on the initial
state of the hidden quantum Markov model.

\subsection{Our contribution}
In this work we propose a Quantum Hidden Markov model formulation using the
notions of Transition Operation Matrices. We focus on Mealy models, for which we
derive first the Forward algorithm in general case, then the Vitterbi algorithm,
for models restricted to those in which sub-TOMs' elements are
trace-monotonicity preserving quantum operations. Subsequently, we discuss the
relationship between our model and model presented in \cite{monras2011hidden}.
The paper ends with the example of application proposed model.

The paper is organised as follows: in Section~\ref{sec:tom} we collect the basic
mathematical objects and their properties, in Section~\ref{sec:qhmm} we define
Quantum Hidden Markov Models and provide Forward na Viterbi algorithm for these
models, in Section~\ref{sec:monras} we discuss the correspondences between
proposed models and models described in \cite{monras2011hidden},
Section~\ref{sec:examples} contains examples of application of our model, and
finally in Section~\ref{sec:conclusions} we conclude.

\section{Transition Operation Matrices}\label{sec:tom}
In what follows we provide basic elements of quantum information theory and
summarize definitions and properties of objects introduced by Gudder in
\cite{gudder2008quantum}.

\subsection{Quantum theory}
Let $\HH$ be a complex finite Hilbert space and $\mathcal{L}(\HH)$ be the set of
linear operators on $\HH$. We also denote the set of positive operators on
$\HH$ as $\mathcal{P^+}(\HH)$ and the set of positive semi-definite operators
on $\HH$ as $\mathcal{P}(\HH)$.

\begin{definition}[Quantum state]
A linear operator $\rho \in \mathcal{P}(\HH)$ is called a quantum state
if $\tr \rho = 1$. Set of quantum states is denoted by $\setofQS(\HH)$.
\end{definition}

\begin{definition}[Sub-normalised quantum state]
A linear operator $\rho \in \mathcal{P}(\HH)$ is called sub-normalised
\cite{cappellini2007subnormalized} quantum state if $\tr \rho \leq 1$. Set of sub-normalised quantum
states is denoted by $\setofQS_\leq(\HH)$.
\end{definition}

\begin{definition}[Positive map]
A linear map $\Phi\in \mathcal{L}(\mathcal{L}(\HH_1), \mathcal{L}(\HH_2))$ is called
positive map if, for every $\rho \in \mathcal{P}(\HH_1)$, $\Phi(\rho) \in
\mathcal{P}(\HH_2)$.
\end{definition}

\begin{definition}[Completely positive map]
A linear map $\Phi\in \mathcal{L}(\mathcal{L}(\HH_1), \mathcal{L}(\HH_2))$ is
called completely positive (CP) if for any complex Hilbert space $\HH_3$,
the map $\Phi \otimes
\1 \in \mathcal{L}(\mathcal{L}(\HH_1 \otimes \HH_3),\mathcal{L}(\HH_2 \otimes \HH_3))$ is positive.
\end{definition}

\begin{definition}[Trace preserving map]
A linear map $\Phi\in \mathcal{L}(\mathcal{L}(\HH_1),\mathcal{L}(\HH_2))$ is
called trace preserving if $\tr(\Phi(\rho)) = \tr \rho$ for every
$\rho \in \mathcal{L}(\HH_1)$.
\end{definition}

\begin{definition}[Trace non-increasing map]
A linear map $\Phi\in \mathcal{L}(\mathcal{L}(\HH_1),\mathcal{L}(\HH_2))$ is
called trace non-increasing if $\tr(\Phi(\rho)) \leq \tr \rho = 1$ for every
quantum state $\rho \in \setofQS(\HH_1)$.
\end{definition}

\begin{definition}[Quantum operation]
    A linear map $\Phi\in \mathcal{L}(\mathcal{L}(\HH_1),\mathcal{L}(\HH_2))$ is
    called a quantum operation if it is completely positive and trace non-increasing.
\end{definition}

\begin{definition}[Quantum channel]
A linear map $\Phi\in \mathcal{L}(\mathcal{L}(\HH_1),\mathcal{L}(\HH_2))$ is
called a quantum channel if it is completely positive and trace preserving.
\end{definition}

\begin{definition}[Quantum measurement]
By quantum measurement we call a mapping from a finite set $\Theta$ of
measurement outcomes to subset of set of measurement operators
$\mu: \Theta \to \mathcal{P}(\HH)$
such that $\sum\limits_{a\in \Theta} \mu(a)=\1$.
\end{definition}

With each measurement $\mu$ we associate a non-negative functional
$p: \Theta \to \R_+\cup\{0\}$ which maps measurement outcome $a$
for a given positive operator $\rho$ and measurement $\mu$ to non-negative real
number in the following way $p(a)_\rho=\tr \mu(a)\rho$.
If $\tr \rho=1$, for given $\rho$ and $\mu$ the value of $p$
can be interpreted as probability of obtaining measurement outcome $a$ in
quantum state $\rho$.

If $\rho$ is a sub-normalised state the trivial measurement $\mu:{a_e}\mapsto
\1$ measures the probability $p(a_e)_\rho=\tr \rho$ that the state $\rho$
exists. One should note that this kind of measurement commutes with any other
measurement and thus does not disturb the quantum system.

\subsection{Transition Operation Matrices}
The core object of the Gudder's scheme is Transition Operation Matrix (TOM)
which generalizes the idea of stochastic matrix.

\begin{definition}[Transition Operation Matrix]
Let $\HH_1$, $\HH_2$ denote two finite dimensional Hilbert spaces and
$\setofQS(\HH_1), \setofQS(\HH_2)$ denote sets of quantum states acting
on those spaces respectively.

A TOM is a matrix in form
$\TOM{E}=\{\TOM{E}_{ij}\}_{i,j=1}^{M,N}$,
where
$\TOM{E}_{ij}$ is completely positive map in $\mathcal{L}(\mathcal{L}(\HH_1),\mathcal{L}(\HH_2))$
such that for every $j$ and $\rho \in \setofQS(\HH_1)$
$\sum_{i}\TOM{E}_{ij}(\rho)\in \setofQS(\HH_2)$.
\end{definition}

Alternatively one can say that
$\TOM{E}=\{\TOM{E}_{ij}\}_{i,j=1}^{M,N}$ is a
TOM if and only if for every column $j$ $\sum_i \TOM{E}_{ij}$ is a~quantum
channel (Completely Positive Trace Preserving map).
A simple implication of this definition is that each $\TOM{E}_{ij}$ is CP-TNI
mapping.

Note that in this definition TOM has four parameters:
\begin{itemize}
\item size of matrix ``output'' (number of rows) --- $M$,
\item size of matrix ``input'' (number of columns) --- $N$,
\item ``input'' Hilbert space --- $\HH_1$,
\item ``output'' Hilbert space --- $\HH_2$.
\end{itemize}
The set of TOMs we will denote as $\Gamma^{M,N}(\HH_1,\HH_2)$.

\begin{definition}[Sub Transition Operation Matrix]
Let $\HH_1$, $\HH_2$ denote two finite dimensional Hilbert spaces,
$\setofQS(\HH_1)$ denotes set of quantum states acting on the first space and $\setofQS^\leq(\HH_2)$
denotes set of sub-normalised quantum states acting on the second Hilbert space.

A sub-TOM is a matrix in form $\TOM{E}=\{\TOM{E}_{ij}\}_{i,j=1}^{M,N}$,
where
$\TOM{E}_{ij}$ is completely positive map in $\mathcal{L}(\mathcal{L}(\HH_1),\mathcal{L}(\HH_2))$
such that for every $j$ and $\rho \in \setofQS(\HH_1)$.
$\sum_{i}\TOM{E}_{ij}(\rho)\in \setofQS^\leq(\HH_2)$.
\end{definition}
The set of sub-TOMs we will denote as $\Gamma_\leq^{M,N}(\HH_1,\HH_2)$.

\begin{definition}[Quantum Markov chain]
Let a TOM $\TOM{E}=\{\TOM{E}_{ij}\}_{i,j=1}^{M,N}$ be given.
Quantum Markov chain is a finite directed graph $G=(E,V)$ labelled by
$\TOM{E}_{ij}$ for $e\in E$ and by zero operator for $e \notin E$.
\end{definition}

\begin{definition}[Vector state]
Vector state is a column vector
$\alpha=[\alpha_1,\alpha_2,\ldots,\alpha_N]^T$ such
that $\alpha_i\in\setofQS_\leq(\HH)$ are sub-normalised quantum states and
$\sum_{i=1}^{N} \alpha_i \in \setofQS(\HH)$.
We will denote the set of vector states as $\Delta^N(\HH)$.
\end{definition}

\begin{definition}[Subnormalised vector state]
Subnormalised vector state is a column vector
$\alpha=[\alpha_1,\alpha_2,\ldots,\alpha_N]^T$ such
that $\alpha_i\in\setofQS_\leq(\HH)$ are sub-normalised quantum states and
$\sum_{i=1}^{N} \alpha_i \in \setofQS_\leq(\HH)$.
W will denote a set of sub-normalised vector states as $\Delta_\leq^N(\HH)$.
\end{definition}

\begin{theorem}[Gudder \cite{gudder2008quantum}]
\label{thm:vs-tom-vs}
Applying TOM $\TOM{E}\in\Gamma^{M,N}(\HH_1,\HH_2)$ on a vector state 
$\alpha\in\Delta^{N}(\HH_1)$ produces vector state
$\beta=\TOM{E}(\alpha)\in \Delta^{M}(\HH_2)$
where $\alpha=[\alpha_1,\alpha_2,\ldots, \alpha_N]^T$,
$\alpha_i\in\setofQS^\leq(\HH_1)$,
where $\beta=[\beta_1,\beta_2,\ldots, \beta_M]^T$,
$\beta_i\in\setofQS^\leq(\HH_2)$,
and
$\TOM{E}\in\Gamma^{M,N}(\HH_1,\HH_2)$, and
in the following way $\beta_i=\sum_{j=1}^{N}\TOM{E}_{ij}(\alpha_j)$.
\end{theorem}

\begin{theorem}[Gudder \cite{gudder2008quantum}]
\label{th:toms-product}
    Product of TOM $\TOM{A} \in \Gamma^{M,N}(\HH_1,\HH_2)$ and $\TOM{B} \in
     \Gamma^{N,K}(\HH_2,\HH_3)$ is a TOM $ \Gamma^{M,K}(\HH_1,\HH_3) \ni \TOM{C} = \TOM{B} \TOM{A}.$
\end{theorem}

\begin{lemma}[Product of two sub-TOMs is a sub-TOM]
    \label{lem:prod-sub-toms}
    Product of sub-TOMs $\TOM{A} \in \Gamma_\leq^{M,N}(\HH_1,\HH_2)$ and $\TOM{B} \in
    \Gamma_\leq^{N,K}(\HH_2,\HH_3)$ is a sub-TOM $ \Gamma_\leq^{M,K}(\HH_1,\HH_3) \ni
    \TOM{C} = \TOM{B} \TOM{A}.$
\end{lemma}

\begin{proof}[Lemma \ref{lem:prod-sub-toms}]
According to proof of Lemma 2.2 in \cite{gudder2008quantum},
$\TOM{C}_{ij}=\TOM{B}_{ij}\TOM{A}_{ij}$ is a completely positive map.
For every $\rho \in \setofQS(\HH_1)$ and $j$ we have that 
$\sigma=\sum_{i=1}^M\TOM{A}_{ij}(\rho)\in\setofQS^\leq(\HH_2)$. 
If $\tr(\sigma)>0$ then  $\tilde{\sigma}=\sigma/\tr(\sigma)\in\setofQS(\HH_2)$ and
\begin{equation}
\begin{split}
\tr\Big(\sum_{i=1}^M \TOM{B}_{ij} (\sigma) \Big)
=&\tr\Big(\tr(\sigma)\sum_{i=1}^M \TOM{B}_{ij} (\tilde{\sigma}) \Big)\\
=&\tr(\sigma) \tr\Big(\sum_{i=1}^M \TOM{B}_{ij} (\tilde{\sigma}) \Big) \leq 1.
\end{split}
\end{equation}
In the case where $\tr(\sigma)=0$, the $\sigma$ is the zero operator and 
$\sum_{i=1}^M \TOM{B}_{ij} (\sigma)$ is also the zero operator. 
Thus $\tr\Big(\sum_{i=1}^M \TOM{B}_{ij} (\sigma) \Big)=0$.
Hence, $\sum_{i=1}^M\TOM{C}_{i,j}(\rho)\in\setofQS^\leq(\HH_3)$ and 
$\TOM{C} \in \Gamma_\leq^{M,K}(\HH_1,\HH_3)$.
\end{proof}

Product of (sub-)TOMs that have same dimensions is associative i.e.
$(\TOM{E}\TOM{F})\TOM{G}=\TOM{E}(\TOM{F}\TOM{G})$ and
$(\TOM{E}\TOM{F})(\alpha)=\TOM{E}(\TOM{F}(\alpha))$.

\section{Mealy Quantum Hidden Markov Model}\label{sec:qhmm}
In order to explain the idea of QHMM we can form following analogy. A QHMM might
by understood as a system consisting of a particle that has an internal
sub-normalised quantum state $\rho\in\setofQS_\leq(\HH)$ and it occupies a
classical state $S_i$. This particle hops from one classical state $S_i$ into
another state $S_j$ passing trough a quantum operation associated with a sub-TOM
element $\TOMel{P}^{V_k}_{S_j, S_i}$. With each transition a symbol $V_k$ is
emitted from the system.

We will now define the classical and quantum version of the Mealy Hidden Markov Model.

\begin{definition}[Finite sequences]
	Let $$\mathcal{V}^T=\underbrace{\mathcal{V}\times \mathcal{V}\times \ldots\times \mathcal{V}}_T$$ 
    defines the set of sequences of length $T$ over alphabet $\mathcal{V}$.
\end{definition}
\begin{definition}[Mealy  Hidden Markov Model]
Let $\mathcal{S} =  \{S_1, \ldots, S_N\}$ and $\mathcal{V} = \{V_1, \ldots,
V_M\}$ be a set of states and an alphabet respectively. The Mealy HMM is
specified by a tuple
$\lambda=(\mathcal{S},\mathcal{V}, \Pi, \pi)$, where:
\begin{itemize}
	\item  $\pi \in [0,1]^{N}$  is a stochastic vector representing initial states, where $\pi_i$ is the
    probability that the initial state is $S_i$;
	\item $\Pi$ is a mapping $\mathcal{V} \ni V_i \mapsto \Pi^{V_i} \in  \mathbb{R}^{N,N}$, 
    where $\Pi^{V_i}$ is sub-stochastic matrix, such that
${\Pi}^{\Sigma}:=\sum\limits_{i=1}^{M} \Pi^{V_i}\in \mathbb{R}^{N,N}$
is stochastic matrix 
and $\Pi^{V_i}_{j,k}$ is $p(n_{t+1}= S_k, o_{t+1} = V_i |
n_t = S_j)$, that is probability of going from state $j$ to $k$ while
generating the output $V_i$.
\end{itemize}
\end{definition}
Let $O = o_1 o_2, \ldots o_T \in \mathcal{V}^T$ be a sequence of length $T$ and $P : \mathcal{V}^T \to [0,1]$ be string probabilities, defined as $P(O) = p(O(1) =
o_1, O(2) = o_2, \ldots, O(T) = o_T)$. The concatenation of
string $O$ and $o_{T+1}$ is denoted by $Oo_{T+1}$.

It is well known that for HMMs the function $P$ satisfies 
\begin{itemize}
	\item $\sum_{O\in\mathcal{V}^T}P(O)=1$ and 
	\item $\sum_{o_{T+1}\in \mathcal{V}} P(O o_{T+1}) = P(O)$, which follows from the law of total probability.
\end{itemize}

The string probabilities generated by Mealy HMM
$\lambda = (\mathcal{S},\mathcal{V},\Pi,\pi)$ are given by
$$P(O|\lambda) =  \sum\limits_{i=1}^{N} \alpha_i,$$
where $\alpha_i$ is $i$-th element of $\alpha=\Pi^{o_T} \Pi^{o_{T-1}} \ldots \Pi^{o_1} \pi$.

\begin{definition}[Mealy Quantum Hidden Markov Model]
Let $\mathcal{S}$ and $\mathcal{V}$ be a set of states and an alphabet
respectively. Mealy QHMM is specified by a tuple
$\lambda=(\mathcal{S},\mathcal{V},\TOM{P},\pi)$, where:
\begin{itemize}
	\item $\pi\in \Delta^N(\HH)$ is an initial vector state;
	\item $\TOM{P}$ is a mapping $\mathcal{V} \to \Gamma_\leq^{N,N}(\HH,\HH)$ such that
$\TOM{P}^S:=\sum\limits_{V_i\in\mathcal{V}} \TOM{P}^{V_i}\in\Gamma^{N,N}(\HH,\HH)$
is a TOM, with $\TOM{P}^{V_i}$ being value of $\TOM{P}$ for $V_i$.
\end{itemize}
\end{definition}
As an example we give a three-state two-symbol Mealy QHMM
$\lambda=(\mathcal{S},\mathcal{V}, \Pi, \pi)$, with
    \begin{equation*}
    \begin{split}
    \mathcal{S}& = \{S_1,S_2,S_3\},\\
    \mathcal{V}& = \{V_1,V_2\},\\
    \Pi &= \left\{
        V_1\mapsto\TOM{P}^{V_1},
        V_2\mapsto\TOM{P}^{V_2}
        \right\},\\
        \pi &=
        \begin{bmatrix}
        \pi_{S_1} \\
        \pi_{S_2} \\
        \pi_{S_3}
        \end{bmatrix},\\
    \TOM{P}^{V_1}&=
    \begin{bmatrix}
    \TOMel{P}_{S_1 S_1}^{V_1} & \TOMel{P}_{S_1 S_2}^{V_1} & \TOMel{P}_{S_1 S_3}^{V_1} \\
    \TOMel{P}_{S_2 S_1}^{V_1} & \TOMel{P}_{S_2 S_2}^{V_1} & \TOMel{P}_{S_2 S_3}^{V_1} \\
    \TOMel{P}_{S_3 S_1}^{V_1} & \TOMel{P}_{S_3 S_2}^{V_1} & \TOMel{P}_{S_3 S_3}^{V_1}
    \end{bmatrix},\\
    \TOM{P}^{V_2}&=
    \begin{bmatrix}
    \TOMel{P}_{S_1 S_1}^{V_2} & \TOMel{P}_{S_1 S_2}^{V_2} & \TOMel{P}_{S_1 S_3}^{V_2} \\
    \TOMel{P}_{S_2 S_1}^{V_2} & \TOMel{P}_{S_2 S_2}^{V_2} & \TOMel{P}_{S_2 S_3}^{V_2} \\
    \TOMel{P}_{S_3 S_1}^{V_2} & \TOMel{P}_{S_3 S_2}^{V_2} & \TOMel{P}_{S_3 S_3}^{V_2}
    \end{bmatrix}.
    \end{split}
    \end{equation*}
A graphical representation of this QHMM is presented
in Fig.~\ref{fig:meallyrepresentation}.

\begin{figure}[ht!]
\begin{center}
    \small
    \begin{tikzpicture}[->,>=stealth',shorten >=1pt,auto,node distance=2.8cm,
    semithick]
    \tikzstyle{every state}=[fill=white,draw=black,thick,text=black,scale=1]
    \node[state]         (S1) at (0,0) {$S_1$};
    \node[state]         (S2) at (3.5,0) {$S_2$};
    \node[state]         (S3) at (1.75,-3.5) {$S_3$};
    \path (S1) edge  [loop above] node[above]
    {$\TOMel{P}_{S_1 S_1}^{V_1}|V_1$} (S1);
    \path (S1) edge  [loop left]  node[left]
    {$\TOMel{P}_{S_1 S_1}^{V_2}|V_2$} (S1);
    \path (S1) edge  [bend right=10] node[sloped,above]
    {$\TOMel{P}_{S_2 S_1}^{V_1}|V_1$} (S2);
    \path (S1) edge  [bend right=40] node[sloped,above]
    {$\TOMel{P}_{S_2 S_1}^{V_2}|V_2$} (S2);
    \path (S1) edge  [bend right=50] node[sloped,above]
    {$\TOMel{P}_{S_3 S_1}^{V_1}|V_1$} (S3);
    \path (S1) edge  [bend right=70] node[sloped,below]
    {$\TOMel{P}_{S_3 S_1}^{V_2}|V_2$} (S3);
    \path (S2) edge  [bend right=50] node[sloped,above]
    {$\TOMel{P}_{S_1 S_2}^{V_1}|V_1$} (S1);
    \path (S2) edge  [bend right=19] node[sloped,above]
    {$\TOMel{P}_{S_1 S_2}^{V_2}|V_2$} (S1);
    \path (S2) edge  [loop above] node[above]
    {$\TOMel{P}_{S_2 S_2}^{V_1}|V_1$} (S2);
    \path (S2) edge  [loop right] node[right]
    {$\TOMel{P}_{S_2 S_2}^{V_2}|V_2$} (S2);
    \path (S2) edge  [bend left=50] node[sloped,above]
    {$\TOMel{P}_{S_3 S_2}^{V_1}|V_1$} (S3);
    \path (S2) edge  [bend left=70] node[sloped,below]
    {$\TOMel{P}_{S_3 S_2}^{V_2}|V_2$} (S3);
    \path (S3) edge  [bend left=20] node[sloped,above]
    {$\TOMel{P}_{S_1 S_3}^{V_1}|V_1$} (S1);
    \path (S3) edge  [bend right=10] node[sloped,above]
    {$\TOMel{P}_{S_1 S_3}^{V_2}|V_2$} (S1);
    \path (S3) edge  [bend left=7] node[sloped,above]
    {$\TOMel{P}_{S_2 S_3}^{V_1}|V_1$} (S2);
    \path (S3) edge  [bend right=20] node[sloped,above]
    {$\TOMel{P}_{S_2 S_3}^{V_2}|V_2$} (S2);
    \path (S3) edge  [out=180+30,in=180+60,looseness=12.] node[left]
    {$\TOMel{P}_{S_3 S_3}^{V_1}|V_1$} (S3);
    \path (S3) edge  [out=-30,in=-60,looseness=12] node[right]
    {$\TOMel{P}_{S_3 S_3}^{V_2}|V_2$} (S3);
    \end{tikzpicture}
\end{center}
\caption{Graphical representation of three-state Mealy QHMM $\lambda$, whose
alphabet consists of two symbols $V_1, V_2$. The symbol $\TOMel{P}_{S_2
S_3}^{V_1}|V_1$ should be understood in the following way: when QHMM is in state
$S_3$ and is being transformed to state $S_2$ while emitting symbol $V_1$, then
the internal quantum sub-state is transformed by quantum operation
$\TOMel{P}_{S_2 S_3}^{V_1}|V_1$.}
\label{fig:meallyrepresentation}
\end{figure}
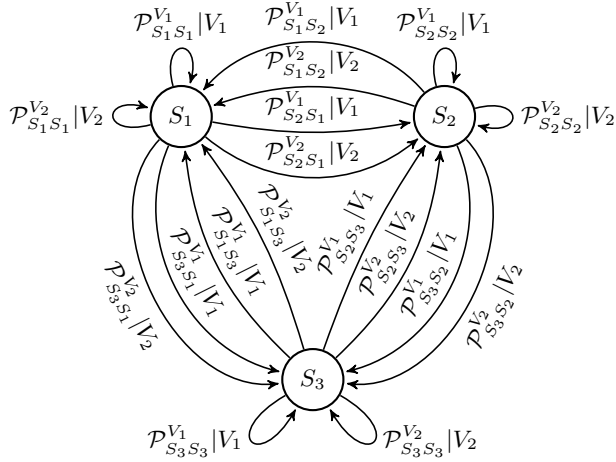

\begin{remark}
For $\dim \HH=1$ QHMM reduces to classical HMM. In this case TOMs reduce to
stochastic matrices, sub-TOMs to sub-stochastic matrices, the
vector states to probability vectors, sub-vector states to sub-normalised
probability vectors.
\end{remark}

\subsection{Forward algorithm for Mealy QHMM}

With each Mealy QHMM we can associate a mapping
$$\varrho : \mathcal{V}^*\to \setofQS_\leq(\HH),$$
where $\mathcal{V}^*= \bigcup_{T=1}^{\infty} \mathcal{V}^T$. Given a sequence $O=(o_1,o_2,\ldots,o_T)$ and a Mealy QHMM $\lambda$ one can compute
resulting sub-normalised quantum state $\rho_{O|\lambda}$.

Let us consider sub-normalised vector states
$$ \alpha_T  = \left[ \alpha_{T,1}, \ldots, \alpha_{T,N}\right]^T  \in \Delta_\leq^N(\HH) $$
such that
\begin{equation}
\label{eq:forward}
\alpha_T=\TOM{P}^{o_T}\ldots\TOM{P}^{o_2}\TOM{P}^{o_1}(\pi),
\end{equation}
then $\rho_{O|\lambda}:=\varrho(O)=\sum\limits_{i=1}^{N}\alpha_{T,i}$.

Equation \eqref{eq:forward} we call the Forward algorithm for QHMMs. Note that
the result of this algorithm is a sub-normalised quantum state
$\rho_{O|\lambda}\in \setofQS_\leq(\HH)$. The sum of all those states over all
possible sequences of a given length forms a quantum state, as formulated in the
following theorem.

\begin{theorem}{}\label{prop:sum_of_vector_states_for_symbol}
	 For any QHMM $\lambda$ we have
    $\sum\limits_{O\in \mathcal{V}^T} \rho_{O|\lambda}\in\setofQS(\HH)$.
\end{theorem}

In order to prove this theorem we will first prove the following lemma.
\begin{lemma}
    \label{lem:sum-vs}
	For any QHMM $\lambda$ the following holds
	$$\sum_{O\in\mathcal{V}^T} \TOM{P}^{o_T}\TOM{P}^{o_{T-1}}\ldots\TOM{P}^{o_{1}}(\pi)\in\Delta^N(\HH).$$
\end{lemma}

\begin{proof}{Lemma \ref{lem:sum-vs}}
	We will proceed by induction.
	For case $T = 1$ we have
	$$\beta_1=\sum_{o\in\mathcal{V}}\TOM{P}^o(\pi) = \TOM{P}^S(\pi) \in \Delta^N(\HH).$$
	For case $T = n+1$
	\begin{equation}
		\begin{split}
		\beta_T = & \sum_{O\in\mathcal{V}^{N+1}}\TOM{P}^{o_{N+1}}\TOM{P}^{o_{N}}\ldots\TOM{P}^{o_{1}}(\pi)  = \\
		& \sum_{o_{n+1}\in\mathcal{V}}\TOM{P}^{o_{n+1}}
		\sum_{O\in\mathcal{V}^{N}}\TOM{P}^{o_{N}}\TOM{P}^{o_{N-1}}\ldots\TOM{P}^{o_{1}}(\pi) = \\
		&  \TOM{P}^S
		\underbrace{\sum_{O\in\mathcal{V}^{N}}\TOM{P}^{o_{N}}\TOM{P}^{o_{N-1}}\ldots\TOM{P}^{o_{1}}}_{\TOM{X}}(\pi).
		\end{split}
	\end{equation}
	By inductive hypothesis $\TOM{X}$ is a TOM. $\TOM{P}^S$ is a TOM, therefore $\beta_T\in \Delta^N(\HH).$
\end{proof}

\begin{proof}{Theorem \ref{prop:sum_of_vector_states_for_symbol}}
	\begin{equation}
	\begin{split}
		\sum\limits_{O\in \mathcal{V}^T} \rho_{O|\lambda}  = &  \\
		 = &\sum\limits_{O\in \mathcal{V}^T} \sum\limits_{i=1}^{N}\alpha_{T,i} \\
		 = &\sum\limits_{O\in \mathcal{V}^T} \sum\limits_{i=1}^{N} \left[\TOM{P}^{o_T}\ldots\TOM{P}^{o_2}\TOM{P}^{o_1}(\pi)\right]_i \\
		 = &\sum\limits_{i=1}^{N} \Big[\underbrace{\sum\limits_{O\in \mathcal{V}^T} \TOM{P}^{o_T}\ldots\TOM{P}^{o_2}\TOM{P}^{o_1}}_{\TOM{X}}(\pi)\Big]_i \\
		 = &\sum\limits_{i=1}^{N} \left[\TOM{X}(\pi)\right]_i \\
	\end{split}
	\end{equation}
	Since by Lemma \ref{lem:sum-vs} $\TOM{X}$ is a TOM, therefore $\sum\limits_{O\in \mathcal{V}^T} \rho_{O|\lambda}\in \setofQS(\HH).$ 
\end{proof}

\begin{theorem}{}\label{trm:prob-concat}
Let $O=(o_1,o_2,...,o_T)$ be a sequence of length $T$ and $Oo_{T+1}$ be a concatenation of $O$ and $o_{T+1}$, then for any QHMM $\lambda$ the following holds
\begin{equation}
\sum_{o_{T+1}\in \mathcal{V}} \rho_{O o_{T+1}|\lambda} = \rho_{O|\lambda}.
\end{equation}
\end{theorem}

\begin{proof}{Theorem \ref{trm:prob-concat}}

According to law of total probability for TOMs \cite{gudder2008quantum} we get
\begin{equation}
\begin{split}
\sum_{o_{t+1}\in \mathcal{V}} \rho_{O o_{t+1}|\lambda} =& \sum_{o_{t+1}\in \mathcal{V}} \sum_{i=1}^N \TOM{P}^{o_{t+1}}_i (\alpha_{T,i})\\
=&\sum_{i=1}^N \sum_{o_{t+1}\in \mathcal{V}} \TOM{P}^{o_{t+1}}_i (\alpha_{T,i})\\
=&\sum_{i=1}^N  \alpha_{T,i} = \rho_{O|\lambda}.
\end{split}
\end{equation}
\end{proof}

\subsection{Viterbi algorithm for Mealy QHMM}
We are given a QHMM $\lambda$ with set of states $\mathcal{S} = \{S_1,
S_2,\dots,S_{N}\}$ and an alphabet of symbols $\mathcal{V} = \{V_1,
V_2, \dots,V_{M}\}$. We will denote
$\TOMel{P}_{ij}^{k}:=\TOM{P}^{V_k}_{S_i S_j}$.

We have a sequence of length $T$, $O=\left(o_1,o_2,\ldots,o_T\right)$, of
symbols from alphabet $\mathcal{V}$, $o_i\in\mathcal{V}$. 

A Mealy QHMM emits symbols on transition from one state to the next. For our
sequence $O$ we index corresponding QHMM states by $n_i$, i.e. $n_0$ is the
initial state (before the emission of the first symbol), and $n_i, i\geq1$ is
the state after emission of the symbol $o_i$. $n_i\in \mathcal{S}$.

The goal of the algorithm is to find most likely sequence of states conditioned
on a sequence of emitted symbols $O$.

We denote the set of partial sequences of state indexes as
$N_k=\{\left(n_0, n_1, \dots, n_k\right):n_j\in\mathcal{S}, j=0, 1, \dots, k\}$,
where $k \leq T$. A set beginning with $n_0$ and ending after $k$ steps with
$S_i$ we denote $N_k^{S_i}=\{\left(n_0, n_1, \dots, n_{k-1},
n_k=S_i\right):n_j\in\mathcal{S}, j=0, 1, \dots, k\}\subset N_k$.

\begin{theorem}{}\label{th:theorem}
Let $O$ be a given sequence of emissions from $\mathcal{V}$. 
Let $\lambda=(\mathcal{S},\mathcal{V},\TOM{P},\pi)$ be a Mealy QHMM satisfying
\begin{equation}\label{th:condition}
\begin{split}
&\forall_{n_i,n_j\in \mathcal{S}, o \in O}
\forall_{\alpha, \beta \in \setofQS_\leq(\HH)} \tr \alpha > \tr \beta \implies\\
&\implies \tr \TOM{P}^o_{n_i,n_j}(\alpha) > \tr \TOM{P}^o_{n_i,n_j}(\beta)
\end{split}
\end{equation}
i.e. all sub-TOMs elements are trace-monotonicity preserving quantum operations.

We define $w\in N_k^{S_i}$ to be a sequence of $k$ states ending with $S_i$. A
sub-normalised state associated with $w$ and sequence $O$ is $B_w \in
\setofQS_\leq(\HH)$ defined as $B_w = \TOM{P}^{o_k}_{n_{k}, n_{k-1}}
	\TOM{P}^{o_{k-1}}_{n_{k-1}, n_{k-2}}
	\ldots \TOM{P}^{o_1}_{n_{1}, n_{0}} (
	\pi_{n_0})$.
The sub-normalised state that maximizes trace over set of all $B_w$s with $w \in N_k^{S_i}$ is
\begin{equation}
A_{k, S_i} = \mathop{\mathrm{argmax}}_{ \left\{B_w:w \in N_k^{S_i} \right\} } \tr B_w.
\end{equation}

Then the following holds
\begin{equation}
\tr A_{k, S_i}
=
\max_{n_{k-1} \in \mathcal{S}}
\tr \TOM{P}^{o_k}_{ n_k=S_i, n_{k-1}}( A_{k-1, n_{k-1}}).
\end{equation}
\end{theorem}

\begin{proof}{}
Let us denote
\begin{equation}
w^*_{k, S_i} = (n^*_0, \ldots, n^*_{k-1}, n^*_k = S_i) \in N_k^{S_i}
\end{equation}
as the sequence of states maximizing trace of $B_w$, so that
\begin{equation}
 \tr A_{k, S_i} = \tr B_{w^*_{k, S_i}}.
\end{equation}

We now have
\begin{equation}
\begin{split}
&\tr A_{k, S_i} =
\max_{w \in N_k^{S_i}} \tr B_w = \\
&=\max_{n_0, \ldots, n_{k-1}, n_k = S_i}
\tr \TOM{P}^{o_k}_{n_k, n_{k-1}}
	\TOM{P}^{o_{k-1}}_{n_{k-1}, n_{k-2}}  \ldots
	\TOM{P}^{o_1}_{n_{1}, n_{0}} (\pi_{n_0})
\end{split}
\end{equation}

Obviously
\begin{equation}
\tr A_{k, S_i} =
\tr \TOM{P}^{o_k}_{n^*_k, n^*_{k-1}}
	\TOM{P}^{o_{k-1}}_{n^*_{k-1}, n^*_{k-2}} \ldots
	\TOM{P}^{o_1}_{n^*_{1}, n^*_{0}} (\pi_{n^*_0})
\end{equation}

We will now prove that for $n^*_k = S_i$
\begin{equation}\label{th:implication}
\begin{split}
& w^*_{k, S_i} = (n^*_0, \ldots, n^*_{k-1}, n^*_k)
\implies\\
& \implies
w^*_{k-1, n^*_{k-1}} = (n^*_0, \ldots, n^*_{k-1}).
\end{split}
\end{equation}

Let us assume that it is not true. That would mean that
$$w^*_{k-1, n^*_{k-1}} =
(l^*_0, \ldots, l^*_{k-2}, n^*_{k-1})
\neq
(n^*_0, \ldots, n^*_{k-2}, n^*_{k-1}).$$
Of course
$$\tr B_{(l^*_0, \ldots, l^*_{k-2}, n^*_{k-1})}
>
\tr B_{(n^*_0, \ldots, n^*_{k-2}, n^*_{k-1})}$$
From this, and (\ref{th:condition}), we have
\begin{equation*}
\begin{split}
& \forall_{n_k\in\mathcal{S}, y_k \in \{1, 2, \ldots, M\}}
\tr \TOM{P}^{o_k}_{n_k, n^*_{k-1}} (B_{(l^*_0, \ldots, l^*_{k-2}, n^*_{k-1})})
> \\
& >
\tr \TOM{P}^{o_k}_{n_k, n^*_{k-1}} (B_{(n^*_0, \ldots, n^*_{k-2}, n^*_{k-1})}),
\end{split}
\end{equation*}
that leads to
$$w^*_{k, S_i} \neq (n^*_0, \ldots, n^*_{k-1}, n^*_k)$$
which is a contradiction. That proves that implication \eqref{th:implication} holds.

Then, for $n_k = S_i$
\begin{equation}
\begin{split}
&\tr A_{k, S_i}
=
\tr B_{w^*_{k, S_i}}
= \\
& =
\tr \TOM{P}^{o_k}_{n^*_k, n^*_{k-1}} (B_{w^*_{k, n^*_{k-1}}})
= \\
& =
\tr \TOM{P}^{o_k}_{n^*_k, n^*_{k-1}} (A_{k - 1, n^*_{k-1}})
= \\
& =
\max_{n_{k-1} \in \mathcal{S} } \tr \TOM{P}^{o_k}_{n_k = S_i, n_{k-1}}(A_{k - 1, n_{k-1}}).
\end{split}
\end{equation}
\end{proof}
\begin{remark}{}
It can be easily seen that Theorem \ref{th:condition} holds iff quantum operation
$\TOMel{P}^y_{n_j,n_i}$ is of form $c\cdot\Phi$, where $c\in[0,1)$ and $\Phi$ is a
quantum channel (CP-TP map).
\end{remark}

From Theorem \ref{th:theorem} we immediately derive the Viterbi algorithm for
Mealy QHMMs conditioned with \eqref{th:condition} that computes most likely
sequence of states for a given sequence $O$.

Initialization:
\begin{equation}
A_{0, S_i} = \pi_{S_i}
\end{equation}

Computation for step number $k$:
\begin{equation}
\begin{split}
&\forall_{S_i \in \mathcal{S}, k\in\{1, \dots, T\}}\ n_{k-1}^*(S_i) = \\
&= \mathop{\mathrm{argmax}}_{n_{k-1} \in \mathcal{S}} \tr \TOM{P}^{o_k}_{S_i, n_{k-1}}(A_{k-1, n_{k-1}})
\end{split}
\end{equation}

\begin{equation}
\forall_{S_i \in \mathcal{S}, k\in\{1, \dots, T\}}\ A_{k, S_i} = \TOM{P}^{o_k}_{S_i, n^*_{k-1}(S_i)}(A_{k - 1, n^*_{k-1}(S_i)}),
\end{equation}

Termination:

\begin{equation}
n_{T}^* = \mathop{\mathrm{argmax}}_{S_i \in \mathcal{S}} \tr A_{T, S_i}.
\end{equation}

The most probable state sequence is $(n_0^*,\dots,n_T^*)$, with resulting state
being $A_{T, n_{T}^*}$ with probability given by $\tr A_{T, n_{T}^*}$.

In case when \eqref{th:condition} does not apply, one can resort to exhaustive
search over all state sequences. As a result of the multitude of possible
quantum operations the behaviour of the Quantum Hidden Markov Model can be
markedly different than its classical counterpart. This is similar to the
relation of quantum and classical Markov models \cite{pawela2015generalized}.

\section{Relation with model proposed by Monras et al.}
\label{sec:monras}
In \cite{monras2011hidden} hidden quantum Markov model is defined, by Monras
et al., as  a tuple
consisting of: a~$d$-level quantum system with an initial state $\rho_0$,
alphabet $\mathcal{V}=\{V_i\}$, a~set of quantum operations (CP-TNI maps)
$\{\mathcal{K}^{V_i}\}$ such that $\sum_i \mathcal{K}^{V_i}$ is 
a~quantum channel (CP-TP map). The system evolves 
in discrete time steps and subsequently generates symbols
$O=\left(o_1,o_2,\ldots,o_T\right)$
from alphabet $\mathcal{V}$
with probability $P(o_t)=\tr \mathcal{K}^{V_i=o_t}(\rho_t)$ in every time step.
After generation of the symbol $o_t$ the
subnormalised quantum state is updated to $\rho_t=\mathcal{K}^{V_i=o_t}
(\rho_{t-1})$. Moreover, $\mathcal{K}^{V_i}$ can be represented by Kraus
operators $\{K^{V_i}_{j}\}$. It means, that $\mathcal{K}^{V_i}(\rho)=\sum_j
K^{V_i}_{j} \rho (K^{V_i}_{j})^\dagger$ and $\rho_t= 
\sum_{j}K^{V_i=o_t}_{j}\rho_{t-1} (K^{V_i=o_t}_{j})^\dagger$, where
$\sum_{j}(K^{V_i}_{j})^\dagger (K^{V_i}_{j})\leq\1$. Here we omit the
normalization factor, therefore with every sequence $O$ a subnormalised quantum
state is associated. 

In the case of Monras et al. model, the number of internal states is
equal to dimension of quantum system. In our case, the states are divided into two
distinct classes: `internal' quantum states and `external' classical states. Our model can
be reduced to the model presented by Monras et al. by performing
the following transformation. First, we need to extend the
alphabet $\mathcal{V}$ with the symbol $\$$. Second, we concatenate every
sequence $O$  with the symbol $\$$. Third, we associate symbol $\$$ with
operation of partial trace over the classical system: 
$\mathcal{K}^{\$}(\rho)=\tr_{\HH_2} \rho$. 
Fourth, we express (sub-)vector states $\alpha$ as block diagonal
(sub-)normalised quantum states and sub-TOMs $\TOM{P}$ as quantum operations. 

According to the above, we can notice, that $\mathcal{K}^{V_i}$ corresponds to
$\mathcal{P}^{V_i}$, whose elements $\mathcal{P}_{k,l}^{V_i}$ are represented by
Kraus operators
$\{E^{V_i}_{k,l,j}\}$, hence $\mathcal{P}_{k,l}^{V_i}(\rho)=\sum_j
E^{V_i}_{k,l,j} \rho (E^{V_i}_{k,l,j})^\dagger$. Let us construct the set of
operators $\{\hat{E}^{V_i}_{k,l,j}\}$ in the form $\hat{E}^{V_i}_{k,l,j}=
E^{V_i}_{k,l,j}\otimes\ketbra{k}{l}$, then similarly as in
\cite{pawela2015generalized}, it can be proved   that $$\sum_{j,k,l}
(\hat{E}^{V_i}_{k,l,j})^\dagger \hat{E}^{V_i}_{k,l,j}\leq\1.$$ Now, consider
vector
state
$\alpha_T=\TOM{P}^{o_T}\ldots\TOM{P}^{o_2}\TOM{P}^{o_1}(\pi)=[\alpha_1,\alpha_2,...,\alpha_N]^T$
with associated a block
diagonal quantum state $\rho_{\alpha}=\sum_i^N
\alpha_i\otimes\ketbra{i}{i}\in\Omega(\HH_1\otimes\HH_2)$, then
\begin{equation}
\rho_{O\$}=\tr_{\HH_2} \sum_{j,k,l}
\hat{E}^{V_i=o_T}_{k,l,j}\cdots\hat{E}^{V_i=o_2}_{k,l,j}\hat{E}^{V_i=o_1}_{k,l,j}
\rho_\alpha
(\hat{E}^{V_i=o_1}_{k,l,j})^\dagger(\hat{E}^{V_i=o_2}_{k,l,j})^\dagger\cdots(\hat{E}^{V_i=o_T}_{k,l,j})^\dagger.
\end{equation}
Thus, our model can be expressed in the language proposed by Monras et
al. However, formalism proposed in this paper has three notable advantages.
First, it presents a hybrid quantum-classical model similar to the one presented in
\cite{li2015quantum} therefore has similar field of applications. 
Our model intuitively generalizes both classical and
quantum models. Second, this model allows us to propose a generalized version
of Viterbi algorithm. Third, the use of TOM and vector states formalism reduces
the amount of memory required to numerically simulate hybrid quantum-classical
Markov models.  

\section{Examples of application}
\label{sec:examples}
\subsection{Example 1}
Let us consider alphabet $\mathcal{V}=\{a, b, c\}$. We define a set of sequences
$\mathcal{O}\subset \mathcal{V}^T$ of length $T$ and having $O_i=a$ for odd $i$, and
$O_i\in\{b,c\}$ for even $i$, i.e. $aba, abaca, acacabaca$.

Let $T=3$. Our objective is to build a model able to differentiate sequences in
$O$ from all other sequences. In classical case, our model could be given by a~HMM
parametrized by
$\lambda_1^c=(\mathcal{S},\mathcal{V},\Pi,\pi)$, where
\begin{equation}
\begin{split}
	&\mathcal{S}=\{s_1,s_2\},\quad \pi=\begin{bmatrix}
	0\\
	1
	\end{bmatrix},\\
	\Pi=
    \Bigg\{
    \Pi^a=&
	\begin{bmatrix}
	0 & 1\\
	0 & 0
	\end{bmatrix},
	\Pi^b=\begin{bmatrix}
	0 & 0\\
	\frac12 & 0
	\end{bmatrix},
	\Pi^c=\begin{bmatrix}
	0 & 0\\
	\frac12 & 0
	\end{bmatrix}
	\Bigg\}.
\end{split}
\end{equation}
It's obvious that $p(aba|\lambda_1^c)=p(aca|\lambda_1^c)=\frac12$, whereas for
other possible sequences we get $\sum_{O\in\mathcal{V}^T\setminus \mathcal{O}}
p(O|\lambda_1^c)=0$.

If we are interested in further differentiating $aba$ from $aca$, we could
either construct two HMMs, one for each sequence, i.e. for $aba$ parametrized by
$\lambda_2^c=(\mathcal{S},\mathcal{V},\Pi,\pi)$, where
\begin{equation}
\Pi=
\Bigg\{
\Pi^a=
\begin{bmatrix}
0 & 1\\
0 & 0
\end{bmatrix},
\Pi^b=\begin{bmatrix}
0 & 0\\
\frac12 & 0
\end{bmatrix}
\Bigg\}
\end{equation}
and similarly for $aca$, or by building a three-state HMM $\lambda_3^c=(\mathcal{S},\mathcal{V},\Pi,\pi)$ 
\begin{equation}
\begin{split}
&\mathcal{S}=\{s_1,s_2,s_3\},\quad \pi=\begin{bmatrix}
0\\
1\\
0
\end{bmatrix},\\
\Pi=
\Bigg\{
\Pi^a=&
\begin{bmatrix}
0 & 1 & 1\\
0 & 0 & 0\\
0 & 0 & 0\\
\end{bmatrix},
\Pi^b=\begin{bmatrix}
0 & 0 & 0\\
\frac12 & 0 & 0\\
0 & 0 & 0\\
\end{bmatrix},
\Pi^c=\begin{bmatrix}
0 & 0 & 0\\
0 & 0 & 0\\
\frac12 & 0 & 0\\
\end{bmatrix}
\Bigg\}	
\end{split}
\end{equation}
and recognize the sequences---$aba$ from $aca$---based on the output of Vitterbi algorithm.

We can solve the problem of discrimination by using QHMM
$\lambda_1^q=(\mathcal{S},\mathcal{V},\TOM{P},\pi)$, with the following
parameters
\begin{equation}
\begin{split}
&\mathcal{S}=\{s_1,s_2\},\quad \pi=\begin{bmatrix}
0_2\\
\ketbra{0}{0}
\end{bmatrix},\\
\Pi=
\Bigg\{
\TOMel{P}^a=&
\begin{bmatrix}
0_4 & \1_4\\
0_4 & 0_4
\end{bmatrix},
\TOMel{P}^b=\begin{bmatrix}
0_4 & 0_4\\
\frac12\Phi_U & 0_4
\end{bmatrix},
\TOMel{P}^c=\begin{bmatrix}
0_4 & 0_4\\
\frac12\1_4 & 0_4
\end{bmatrix}
\Bigg\},
\end{split}
\end{equation}
where $\Phi_U(\cdot)=U\cdot U^\dagger$ is unitary channel, such that 
$U=
\begin{bmatrix}
\cos \frac{\pi}{2} & -\sin \frac{\pi}{2}\\
\sin \frac{\pi}{2} & \cos \frac{\pi}{2}
\end{bmatrix}
$ 
and $0_n$, $\1_n$ are zero and identity operators over vector space of dimension $n$, respectively.

Moreover, let $\mu:\{b \mapsto \ketbra{1}{1},c \mapsto \ketbra{0}{0}\}$ be a measurement.

Let us consider the application of quantum forward algorithm on sequence $aba$.
Initial vector state of the algorithm is
$
\alpha_0=
\begin{bmatrix}
0_2\\
\ketbra{0}{0}
\end{bmatrix},
$
the final state is 
$\alpha_3=
\begin{bmatrix}
\frac12 U\ketbra{0}{0}U^\dagger\\
0_2
\end{bmatrix}.
$
The associate sub-normalised quantum state is
$\rho_{aba|\lambda^q_1}=\frac12 U\ketbra{0}{0}U^\dagger$, therefore the 
resulting sequence of probabilities is given by 
\begin{equation}
	(\tr\rho_{aba|\lambda^q_1}\mu(b),\tr\rho_{aba|\lambda^q_1}\mu(c))=\left(\frac12,0\right).
\end{equation}
It is obvious that application of quantum forward algorithm on sequence $aca$
gives result $(0,\frac12)$ and $\sum_{O\in\mathcal{V}^T\setminus \mathcal{O}}
\tr\rho_{O|\lambda^q_1}=0$.

We have shown that it is possible to construct two-state QHMM that fulfils the
same task as pair of two-states HMMs or three-state HMM.

\subsection{Example 2}
Let us consider language $A$ consisting of the sequences
$a^{k_1}b^{k_2}a^{k_3}\cdots$, where $k_1,k_2,k_3,... $ are nonnegative odd
integers and $a,b$ are symbols from alphabet $\mathcal{V}=\{a, b\}$. 
In other words language $A$ contains these sentences in which
odd length subsequences of letters $a$ and $b$ alternate.

Classically, sequences from this language can be generated by four-state HMM
$\lambda^c=(\mathcal{S},\mathcal{V},\Pi,\pi)$ presented in Fig.~\ref{fig:CHMM},
where
\begin{equation}
\begin{split}
\mathcal{S}=\{s_1,s_2,s_3,s_4\},\quad &\pi=\begin{bmatrix}
1\\
0\\
0\\
0
\end{bmatrix},\\
\Pi=
\Bigg\{
\Pi^a=
\begin{bmatrix}
0 & 1 & 0 & 0\\
\frac{1}{2} & 0 & 0 & 0\\
0 & 0 & 0 & 0\\
\frac{1}{2} & 0 & 0 & 0\\
\end{bmatrix},
&\Pi^b=\begin{bmatrix}
0 & 0 & 0 & \frac{1}{2}\\
0 & 0 & 0 & 0\\
0 & 0 & 0 & \frac{1}{2}\\
0 & 0 & 1 & 0\\
\end{bmatrix}
\Bigg\}.
\end{split}
\end{equation}
It is easy to check, that for any sequence $a^{k_1}b^{k_2}a^{k_3}\cdots$ from the language $A$, probability $p(a^{k_1}b^{k_2}a^{k_3}\cdots|\lambda^c)$ is nonzero and equals $p(a^{k_1}b^{k_2}a^{k_3}\cdots|\lambda^c)=(\frac{1}{2})^{\frac{k_1+1}{2}}(\frac{1}{2})^{\frac{k_2+1}{2}}(\frac{1}{2})^{\frac{k_3+1}{2}}\cdots$. Moreover, if any $k_i$ is even, then ${p(a^{k_1}b^{k_2}a^{k_3}\cdots|\lambda^c)=0}$.

Let us consider matrix of probabilities $p(a^{k_1}b^{k_2}a^{k_3}\cdots|\lambda^c)$ given as
	\begin{equation}\label{eq:hankelinf}
	\tilde{H}=
	\begin{bmatrix}
	1&p(a|\lambda^c)&p(b|\lambda^c)&p(aa|\lambda^c)&p(ab|\lambda^c)&p(ba|\lambda^c)&\cdots\\
	p(a|\lambda^c)&p(aa|\lambda^c)&p(ba|\lambda^c)&p(aaa|\lambda^c)&p(aba|\lambda^c)&p(baa|\lambda^c)&\cdots\\
	p(b|\lambda^c)&p(ab|\lambda^c)&p(bb|\lambda^c)&p(aab|\lambda^c)&p(abb|\lambda^c)&p(bab|\lambda^c)&\cdots\\
	p(aa|\lambda^c)&p(aaa|\lambda^c)&p(baa|\lambda^c)&p(aaaa|\lambda^c)&p(abaa|\lambda^c)&p(baaa|\lambda^c)&\cdots\\
	p(ab|\lambda^c)&p(aab|\lambda^c)&p(bab|\lambda^c)&p(aaab|\lambda^c)&p(abab|\lambda^c)&p(baab|\lambda^c)&\cdots\\
	p(ba|\lambda^c)&p(aba|\lambda^c)&p(bba|\lambda^c)&p(aaba|\lambda^c)&p(abba|\lambda^c)&p(baba|\lambda^c)&\cdots\\
	\vdots&\vdots&\vdots&\vdots&\vdots&\ddots\\
	\end{bmatrix}.
	\end{equation}

Notice, that any upper-left  corner of matrix $\tilde{H}$ is known as
the Hankel matrix. Denote by $\tilde{H}_d$ a upper-left $d$-size submatrix of
matrix $\tilde{H}$. Subsequently, let us notice that
{
	\renewcommand{\arraystretch}{1.5}
	\setlength{\arraycolsep}{5pt}
	\begin{equation}
	\mathrm{rank}(\tilde{H}_{11})=
	\mathrm{rank}\begin{bmatrix}
	1&\frac{1}{2}&0&\frac{1}{2}&\frac{1}{4}&0&0&\frac{1}{4}&0&\frac{1}{8}&\frac{1}{4}\\
	\frac{1}{2}&\frac{1}{2}&0&\frac{1}{4}&\frac{1}{8}&0&0&\frac{1}{4}&0&\frac{1}{8}&0\\
	0&\frac{1}{4}&0&0&\frac{1}{4}&0&0&\frac{1}{8}&0&\frac{1}{16}&\frac{1}{8}\\
	
	\frac{1}{2}  &   \frac{1}{4} &   0    &  \frac{1}{4} &   \frac{1}{8} &  0   &   0  &    \frac{1}{8} &  0 & \frac{1}{16} & 0\\
	\frac{1}{4} &   0   &   0   &   \frac{1}{8} &  \frac{1}{16} &  0  &    0  &    0   &   0  &    0&
	0\\
	0   &    \frac{1}{8}  &  0   &    0    &   0   &    0  &     0   &    \frac{1}{16} &
	0    &   \frac{1}{32} & \frac{1}{16}\\
	0   &   \frac{1}{4} &   0   &   0    &  \frac{1}{8} &  0   &   0   &   \frac{1}{8} &  0&
	\frac{1}{16} & \frac{1}{8}\\
	\frac{1}{4}  &  \frac{1}{4} &   0  &    \frac{1}{8} &  \frac{1}{16}&  0  &    0   &   \frac{1}{8} &  0&
	\frac{1}{16} & 0\\
	0    &   \frac{1}{8}  &  0    &   0    &   0   &    0   &    0   &    \frac{1}{16} &
	0    &   \frac{1}{32} & 0   &\\
	\frac{1}{8} &   0  &     0  &     \frac{1}{16} &  \frac{1}{32} & 0    &   0 &      0     &  0 &
	0  &     0\\
	\frac{1}{4} &   0   &   0 &     \frac{1}{8} &  \frac{1}{16} & 0  &    0 &     0 &     0 &     0 &
	0    
	\end{bmatrix}=4.
	\end{equation}}
Since $\mathrm{rank}(\tilde{H}_{11})=4$, four-state HMM
$\lambda^c=(\mathcal{S},\mathcal{V},\Pi,\pi)$ cannot be reduced to HMM with
smaller number of states \cite{vidyasagar2011complete,huang2014minimal}.

The application of the QHMM for the generation of sequences from $A$ can reduce
the number of the states to three. Let us consider QHMM
$\lambda^q=(\mathcal{S},\mathcal{V},\Pi,\pi)$ presented in Fig.~\ref{fig:QHMM},
with
\begin{equation}
\begin{split}
\mathcal{S}=\{s_1,s_2,s_3\},\quad &\pi=\begin{bmatrix}
0_2\\
0_2\\
\ketbra{0}{0}
\end{bmatrix},\\
\Pi=
\Bigg\{
\TOMel{P}^a=
\begin{bmatrix}
0_4 & \Phi_{\ketbra{+}{+}} & \Phi_{H\!\ketbra{0}{0}}\\
\Phi_{H\!\ketbra{0}{0}} & 0_4 & 0_4\\
0_4 & 0_4 & 0_4
\end{bmatrix},
&\TOMel{P}^b=\begin{bmatrix}
0_4 & 0_4 & 0_4\\
0_4 & 0_4 & \Phi_{H\!\ketbra{1}{1}}\\
\Phi_{H\!\ketbra{1}{1}} & \Phi_{\ketbra{-}{-}} & 0_4
\end{bmatrix}
\Bigg\},
\end{split}
\end{equation}
where $\Phi_X(\cdot)=X\cdot X^\dagger$ and $X\in\{\ketbra{+}{+},\ketbra{-}{-},H\!\ketbra{0}{0},H\!\ketbra{1}{1}\}$.

\begin{figure}[ht!]
	\centering
	\subfigure[]{%
		\centering
		\small
		\begin{tikzpicture}[->,>=stealth',shorten >=1pt,auto,node distance=2.8cm,
		semithick]
		\tikzstyle{every state}=[fill=white,draw=black,thick,text=black,scale=1]
		\node[state]         (S1) at (0,0) {$s_1$};
		\node[state]         (S2) at (1.5,0) {$s_2$};
		\node[state]         (S3) at (3,0) {$s_3$};
		\node[state]         (S4) at (4.5,0) {$s_4$};
		\path (S1) edge  [bend right=-40] node[sloped,above]
		{$\frac{1}{2}\Big|a$} (S2);
		\path (S1) edge  [bend right=-80] node[sloped,above]
		{$\frac{1}{2}\Big|a$} (S4);
		\path (S2) edge  [bend right=-40] node[sloped,below]
		{$1\Big|a$} (S1);

		\path (S4) edge  [bend right=-40] node[sloped,below]
		{$\frac{1}{2}\Big|b$} (S3);
		\path (S4) edge  [bend right=-80] node[sloped,below]
		{$\frac{1}{2}\Big|b$} (S1);
		\path (S3) edge  [bend right=-40] node[sloped,above]
		{$1\Big|b$} (S4);
		
		\end{tikzpicture}
		\label{fig:CHMM}
	}%
	\subfigure[]
	{%
		\centering
		\begin{tikzpicture}[->,>=stealth',shorten >=1pt,auto,node distance=2.8cm,
		semithick]
		\tikzstyle{every state}=[fill=white,draw=black,thick,text=black,scale=1]
		\node[state]         (S1) at (0,0) {$s_1$};
		\node[state]         (S2) at (2.25,0) {$s_2$};
		\node[state]         (S3) at (4.5,0) {$s_3$};
		
		\path (S1) edge  [bend right=-40] node[sloped,above]
		{$\Phi_{H\ketbra{0}{0}}\Big|a$} (S2);
		\path (S1) edge  [bend right=-80] node[sloped,above]
		{$\Phi_{H\ketbra{1}{1}}\Big|b$} (S3);
		\path (S2) edge  [bend right=-40] node[sloped,below]
		{$\Phi_{\ketbra{+}{+}}\Big|a$} (S1);

		\path (S3) edge  [bend right=-40] node[sloped,below]
		{$\Phi_{H\ketbra{1}{1}}\Big|b$} (S2);
		\path (S3) edge  [bend right=-80] node[sloped,below]
		{$\Phi_{H\ketbra{0}{0}}\Big|a$} (S1);
		\path (S2) edge  [bend right=-40] node[sloped,above]
		{$\Phi_{\ketbra{-}{-}}\Big|b$} (S3);
		
		\end{tikzpicture}
		\label{fig:QHMM}
	}%	
	\caption{Examples of HMM (a) and QHMM (b) generating with nonzero probabilities sequences $a^{k_1}b^{k_2}a^{k_3}\cdots$, where $k_1,k_2,k_3,\ldots $ are nonnegative odd integers.} %
	\label{fig:example2}
\end{figure}
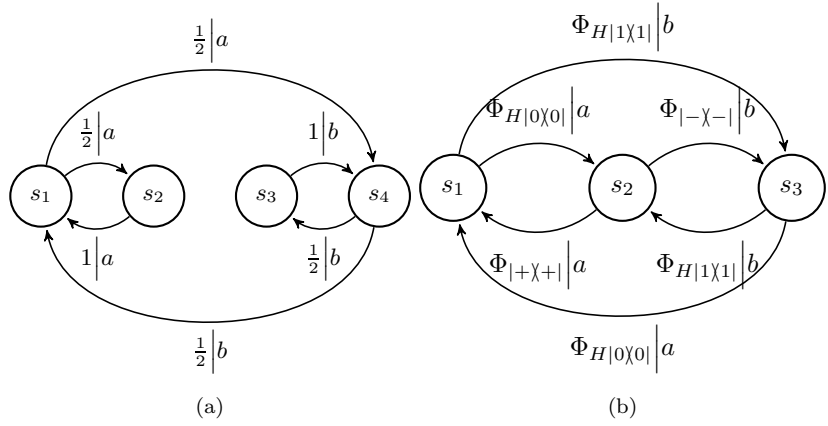
Notice that, for any sequence $a^{k_1}b^{k_2}a^{k_3}\cdots$, where $k_1,k_2,k_3,... $ are nonnegative odd integers, the final state is given as 
$$\alpha_{k_1k_2k_3\dots}=\begin{bmatrix}
(\frac{1}{2})^{\frac{k_1+1}{2}}(\frac{1}{2})^{\frac{k_2+1}{2}}(\frac{1}{2})^{\frac{k_3+1}{2}}\cdots\left[\begin{matrix}1 & 1\\1 & 1 \end{matrix}\right]\\
0_2\\
0_2
\end{bmatrix}$$ or
$$\alpha_{k_1k_2k_3\dots}=\begin{bmatrix}
0_2\\
0_2\\
(\frac{1}{2})^{\frac{k_1+1}{2}}(\frac{1}{2})^{\frac{k_2+1}{2}}(\frac{1}{2})^{\frac{k_3+1}{2}}\cdots\left[\begin{matrix}1 & -1\\-1 & 1 \end{matrix}\right]
\end{bmatrix}.$$ 
Moreover, if any $k_i$ is even, then $\alpha_{k_1k_2k_3\dots}=\begin{bmatrix}
0_2\\
0_2\\
0_2
\end{bmatrix}$. Therefore we have shown, that it is possible to construct
thee-state QHMM generating sequences 
from $A$ with the same probabilities like its classical four-state counterpart.
Those probabilities $\tr\rho_{a^{k_1}b^{k_2}a^{k_3}\cdots|\lambda^q}$ are obtained
from trivial 
measurements of sub-normalised quantum states $\rho_{a^{k_1}b^{k_2}a^{k_3}|\lambda^q\cdots}$.

\section{Conclusions}\label{sec:conclusions}
We have introduced a new model of Quantum Hidden Markov Models based on the
notions of Transition Operation Matrices and Vector States. We have shown that
for a subclass of QHMMs and emission sequences the modified Viterbi algorithm
can be used to calculate the most likely sequence of internal states that lead
to a given emission sequence. Because of the fact that the structure of Quantum
Hidden Markov Models is more complicated than their classical counterparts, in
general case the most likely sequence of states leading to a given emissions
sequence has to be calculated using extensive search. We have also proposed a
formulation of the Forward algorithm that is applicable for general QHMMs.

For given a sequence of symbols of length $T$, $O=(o_1,o_2,...,o_T)$, a sequence
of states $N_T=(n_0,n_1,...n_T)$ and a classical Mealy HMM with parameters $\lambda$,
the joint probability distribution $P(N_T,O)$ can be factored
into 
\begin{equation}
P(N_T,\!O)\!=\!P(n_0)\prod_{t=1}^T\!P(n_t|o_t,n_{t-1})P(o_t|n_{t-1}).
\end{equation}
As in the case of classical Moore HMM \cite{ghahramani2001introduction}, the
above factorization can be considered as a simple dynamic Bayesian Network.
Hence, the concept of QHMM proposed in this manuscript gives basis to quantum
generalization of dynamic Bayesian Networks.

We believe that proposed model can find applications in modelling systems that
posses both quantum and classical features.

\section*{Acknowledgements}
We would like to thank Z. Pucha\l{}a and \L{}. Pawela for fruitful discussions
about subject of this paper.

This paper was partially supported by Polish National Science Centre. 
P.~Gawron was supported by  grant number 2014/15/B/ST6/05204.
P.~G\l{}omb was supported by grant number DEC-2011/03/D/ST6/03753.
M.~Cholewa was supported by grant number  DEC-2012/07/N/ST6/03656.
D.~Kurzyk was supported by grant number UMO-2013/11/N/ST6/03090.

\end{document}